\documentclass[a4paper,UKenglish,cleveref, autoref, thm-restate]{lipics-v2021}

\title{Minimum Sum Vertex Cover via Minimum Vertex Cover} 

\author{Ahmad Biniaz}{School of Computer Science, University of Windsor}{Ahmad.Biniaz@uwindsor.ca}{https://orcid.org/0000-0002-6396-4494}{Research supported by NSERC.}

\author{Jean-Lou De Carufel}{School of Electrical Engineering and Computer Science, University of Ottawa}{jdecaruf@uottawa.ca }{https://orcid.org/0000-0002-6734-8234}{Research supported by NSERC.}

\author{Anil Maheshwari}{School of Computer Science, Carleton University}{anil@scs.carleton.ca}{https://orcid.org/0000-0002-1274-4598}{Research supported by NSERC.}
\author{Saeed Odak}{Department of Computer Science, Aalto University.}{Saeed.Odak@aalto.fi}{https://orcid.org/0009-0005-6290-0965}{}
\author{Michiel Smid}{School of Computer Science, Carleton University}{michiel@scs.carleton.ca }{https://orcid.org/0000-0003-3683-9054}{Research supported by NSERC.}

\authorrunning{Biniaz, De Carufel, Maheshwari, Odak, and Smid} 

\Copyright{Ahmad Biniaz, Jean-Lou De Carufel, Anil Maheshwari, Saeed Odak, and Michiel Smid} 



\category{} 

\relatedversion{} 

\acknowledgements{We thank the anonymous reviewers for their careful reading of the manuscript and for their valuable comments and suggestions.}

\EventEditors{Lin Chen and Nicole Megow}
\EventNoEds{2}
\EventLongTitle{37th International Symposium on Algorithms and Computation (ISAAC 2026)}
\EventShortTitle{ISAAC 2026}
\EventAcronym{ISAAC}
\EventYear{2026}
\EventDate{December 6--9, 2026}
\EventLocation{Hangzhou, China}
\EventLogo{}
\SeriesVolume{399}
\ArticleNo{7}

\usepackage[utf8]{inputenc}
\usepackage[T1]{fontenc}
\usepackage{amsmath, amssymb, amsthm}
\usepackage{mathtools}
\usepackage{hyperref}
\usepackage{enumitem}
\usepackage{booktabs}
\usepackage{comment}
\usepackage{algorithm2e}
\usepackage{algpseudocode}
\usepackage{xcolor}

\DeclareMathOperator{\cost}{\mathrm{cost}}
\newcommand{\SVC}{\operatorname{\overline{svc}}}
\newcommand{\OPT}{\operatorname{\textsc{opt}}}

\newcommand{\rd}{\operatorname{rd}}

\hypersetup{
    colorlinks=true,
    linkcolor=blue,
    citecolor=blue,
    urlcolor=blue
}

\nolinenumbers

\begin{document}
\begin{titlepage}
\maketitle

\begin{abstract}
The Minimum Sum Vertex Cover (MSVC) problem asks for an ordering of the vertices of a graph that minimizes the sum, over all edges, of the time at which each edge is first covered. We study the problem through the structure of vertex covers and obtain new approximation and exact algorithms, together with conditional lower bounds.

For graphs of maximum degree $\Delta$, we show that a simple ordering algorithm based on a minimum vertex cover achieves approximation ratio $R_\Delta \le {(\sqrt{\Delta}+1)}/{2}$. For $d$-regular graphs, we give a polynomial-time $1.184$-approximation by combining Max-$k$-Vertex-Cover approximation with a structural bound on optimal prefixes.

On the exact side, we give an algorithm parameterized by the vertex cover number $k$ running in $2^{O(k\log k)} + O(n+m)$ time, improving the previous dependence on $k$, where $n$ and $m$ are the number of vertices and edges in the graph, respectively. We also develop a separator-based exact algorithm running in $ 2^{O(\sqrt n \log n)}$ time on planar, bounded-genus, and fixed-minor-free graph classes. 

Finally, we prove that Minimum Sum Vertex Cover is NP-hard on planar graphs and, assuming ETH, admits no $2^{o(\sqrt n)}$-time exact algorithm on $n$-vertex planar graphs. Thus our planar upper bound is tight up to logarithmic factors in the exponent.
\end{abstract}

\end{titlepage}

\section{Introduction}

Many optimization problems ask not only which objects should be selected but also in what order. In such problems the objective is governed by latency: requests should be served as early as possible, and the cost of a solution depends on the time at which each request is first satisfied. This point of view appears naturally in scheduling, query optimization, fault diagnosis, information retrieval, and related covering problems~\cite{DBLP:conf/soda/AzarG11, BansalBFT21}.

The \textsc{Minimum Sum Vertex Cover} (MSVC) problem is a basic graph-theoretic model of this phenomenon. Given a graph $G=(V,E)$, the task is to order the vertices so that edges are covered as early as possible. MSVC is the graph special case of \textsc{Minimum Sum Set Cover}, introduced by Feige, Lov\'asz, and Tetali~\cite{FLT04}. Formally, an ordering is a bijection $\sigma:\{1,\ldots,|V|\} \to V$, and an edge $\{u,v\}\in E$ is covered at time $\min\{\sigma^{-1}(u),\sigma^{-1}(v)\}$. The objective is to compute a bijection $\sigma$ that minimizes 
\[
    \cost(\sigma) := \sum_{\{u,v\}\in E}\min\{\sigma^{-1}(u),\sigma^{-1}(v)\}.
\]
We denote the optimum value by $\mathrm{msvc}(G)$. Equivalently, the vertices are revealed one by one, an edge is covered when its first endpoint is revealed, and the goal is to minimize the sum of the cover times of all edges. Any MSVC ordering naturally induces a vertex cover: for each edge, take the endpoint that appears first in the ordering. This connection motivates studying MSVC through the structure of Minimum Vertex Cover.

\medskip\noindent\textbf{Approximation.}
MSVC is NP-hard, and much of the previous work has focused on approximation algorithms. Feige, Lov\'asz, and Tetali~\cite{FLT04} gave a $2$-approximation as part of their work on \textsc{Minimum Sum Set Cover}. Bansal, Batra, Farhadi, and Tetali~\cite{BansalBFT21} improved this to a $16/9$-approximation using a linear-programming relaxation and a tailored rounding scheme; they also showed a matching integrality gap for their relaxation. This remains the best known polynomial-time approximation ratio for general graphs. On the hardness side, Stankovi\'c~\cite{Stankovic25} proved that, assuming the Unique Games Conjecture, MSVC is hard to approximate within a factor of $1.0748$. The same work studied regular graphs, proving hardness below $1.0157$ and giving a $1.225$-approximation based on approximation algorithms for \textsc{Max-$k$-Vertex-Cover}~\cite{AS19,M19,RT12}.
 
\medskip\noindent\textbf{Structured graph classes and trees.}
Beyond worst-case approximation, MSVC has been studied on structured graph classes, where both exact values and combinatorial bounds are known. Gera, Rasmussen, Stanica, and Horton~\cite{gera2006results} determined the exact optimum for several families, including complete bipartite graphs and hypercubes, and gave general upper and lower bounds on the MSVC cost in terms of the independence number, the girth, and the vertex cover number of the graph. Polynomial-time exact algorithms are known for caterpillars and split graphs. However, whether MSVC is polynomial-time solvable or NP-hard on \emph{trees} has remained open for many years~\cite{AP25, rasmussen2006efficient}. This open problem motivates our study of bounded-degree graphs that admit polynomial-time algorithms for the Vertex Cover problem, a class that includes bounded-degree trees.
 
\medskip\noindent\textbf{Parameterized complexity.}
The problem has also been examined through the lens of parameterized complexity, where the choice of parameter is delicate. The optimum value of any connected $n$-vertex graph is $\Omega(n)$, so parameterizing by the solution cost is uninformative. Parameterizing by treewidth is also problematic, since the complexity of the problem on trees, which have treewidth one, remains open. Aute and Panolan~\cite{AP25} gave fixed-parameter algorithms parameterized by the size $k$ of a minimum vertex cover, running in $2^{2^{O(k)}} n^{O(1)}$ time, and by the size of a minimum clique modulator. Liu, Cao, Gai, and Wang~\cite{LCGW24,LCGW25} studied a different parameter, namely the largest position at which an edge is first covered in the ordering. For this parameter $p$, they obtained a polynomial kernel with $p^2+2p$ vertices and an algorithm running in time $O(2^p p! p^3 + m)$ where $m$ is the size of graph.
A subsequent structural result of Zhang and Cao~\cite{DBLP:journals/corr/abs-2605-21920} shows that, for every graph $G$, $p \le 2k\log k$, where $k$ is the size of a minimum vertex cover. They also exhibit a family of bipartite graphs for which $p = \Omega(k \log k / \log\log k)$. Together with the algorithm of~\cite{LCGW24,LCGW25}, this implies a $2^{O(k\log^2 k)} + O(m)$ time algorithm for MSVC parameterized by~$k$, where $m$ is the size of the input graph.
 
\medskip\noindent\textbf{Our results.}
We make progress on MSVC in approximation, exact algorithms, and conditional lower bounds.

\begin{itemize}[leftmargin=1em]
    \item \textbf{Bounded-degree graphs via minimum vertex cover.}
    We give a simple ordering algorithm that starts from a minimum vertex cover and greedily builds a star decomposition of the edges. For graphs of maximum degree $\Delta$, the algorithm has approximation ratio of
    \[
        R_\Delta = \max_{1\le k\le \Delta} \frac{k(\Delta-1)}{\Delta-1+(k-1)^2},
    \]
    and $R_\Delta\le(\sqrt{\Delta}+1)/2$; see \Cref{thm:min-cover-star-ordering}. In particular, the ratio is $4/3$ for $\Delta=3$, giving a polynomial-time $4/3$-approximation for subcubic graph classes (including trees) where a minimum vertex cover can be found in polynomial-time. We also show that the analysis of this minimum-cover ordering is tight in the prescribed-cover setting; see \Cref{thm:min-cover-tight}.

    \item \textbf{Regular graphs.}
    We give a polynomial-time $1.184$-approximation for MSVC on $d$-regular graphs, improving the previous $1.225$ bound of Stankovi\'c~\cite{Stankovic25}; see \Cref{thm:regular}. The algorithm uses a \textsc{Max-$k$-Vertex-Cover} approximation subroutine, but unlike the previous analysis it does not commit to a single value of $k$. Instead, it tries all choices of $k$ and uses~a structural lemma to bound the number of edges covered by prefixes of an optimal~ordering (\Cref{lem:coverage}).

    \item \textbf{Exact FPT algorithm parameterized by vertex cover.}
    We give an exact algorithm for MSVC parameterized by the size $k$ of a vertex cover. The running time is
    \[
        2^{O(k\log k)}+O(n+m),
    \]
    improving the previous know result by a $\log k$ factor in the exponent; see \Cref{thm:msvc-vc-exact}.

    \item \textbf{Exact algorithms via separators.}
    For general graphs, a Held--Karp style subset dynamic program solves MSVC in $2^n n^2$ time and this is qualitatively optimal under Gap-ETH; see \Cref{thm:exact-general-graphs}. We then obtain faster exact algorithms on sparse graph classes with small balanced separators. If every subgraph admits balanced separators whose sizes sum to $\Lambda(n)$ along a recursive decomposition, then MSVC can be solved exactly in time and space
    \[
        2^{O((\Lambda(n)+\sqrt m)\log n)};
    \]
    see \Cref{thm:exact-general}. This yields $2^{O(\sqrt n\log n)}$ algorithms for planar, bounded-genus, and $H$-minor-free graphs, and a corresponding separator-based bound for bounded-treewidth graphs; see \Cref{thm:planar-msvc-exact,cor:exact-treewidth}.

    \item \textbf{Planar hardness and lower bounds.}
    We extend NP-completeness of the decision version of MSVC to planar graphs; see \Cref{thm:np-complete}. Then using the same reduction, we show that unless ETH fails, MSVC has no $2^{o(\sqrt n)}$ time exact algorithm on $n$-vertex planar graphs; see \Cref{thm:eth}. Thus the separator-based upper bound is tight up to logarithmic factors in the exponent.
\end{itemize}
 
\medskip\noindent\textbf{Related vertex-ordering problems.}
MSVC belongs to a broader family of graph ordering problems whose objectives are sums or maxima over edges. Examples include the minimum linear arrangement problem~\cite{DBLP:journals/algorithmica/CharikarHKR10,DBLP:conf/stoc/DevanurKSV06},
the minimum logarithmic arrangement problem~\cite{MestrePupyrev22}, and the bandwidth problem~\cite{DBLP:journals/jcss/DubeyFU11}. We refer to the survey of D\'iaz, Petit, and Serna~\cite{DiazPetitSerna02} for further variants. Many of these problems become more tractable on graph classes with small separators, such as planar graphs, bounded-genus graphs, minor-free graphs, and graphs of bounded treewidth. Our separator algorithm for MSVC fits this general pattern, although the MSVC objective has an important difference: it depends on absolute positions in the ordering, not only on the distance between the endpoints of an edge. Thus techniques for translation-invariant layout objectives do not apply directly.

\section{Preliminaries}\label{sec:prelim}

All graphs in this paper are finite, simple, and undirected. For a graph $G=(V,E)$ we~write $n=|V|$ and $m=|E|$. Isolated vertices do not affect the MSVC cost and may be placed last. For $v\in V$, let $N(v)$ be the neighbourhood of $v$ and let $\deg(v)=|N(v)|$. The maximum~degree is denoted by $\Delta(G)$, and $G$ is $d$-regular if every vertex has degree $d$. For $S\subseteq V$, we write $G[S]$ for the subgraph induced by $S$ and call $S$ independent if $G[S]$ has no edges. A vertex~cover is a set of vertices meeting every edge, and $\tau(G)$ denotes the minimum size of a vertex~cover.

For an ordering $\sigma$ of $V$, we write $\cost(\sigma)$ for its MSVC cost and $\mathrm{msvc}(G)$ for the optimum. For a minimization problem, a $\rho$-approximation is an algorithm that always returns a feasible solution of cost at most $\rho$ times the optimum; for MSVC this means $\cost(\sigma_{\mathrm{ALG}}) \le \rho\cdot \mathrm{msvc}(G)$.

A balanced separator of an $N$-vertex graph is a set $S\subseteq V$ such that $V\setminus S$ can be partitioned into two sets with no edge between them, each of size at most $2N/3$. For any function $\operatorname{sep}(\cdot)$, a graph class has $\operatorname{sep}(\cdot)$ separators if every $N$-vertex graph has a balanced separator of size at most $\operatorname{sep}(N)$. Planar graphs have $O(\sqrt N)$ separators~\cite{LiptonTarjan79}, and the same asymptotic bound holds for every fixed-minor-free class~\cite{AlonSeymourThomas90}. Graphs of treewidth $t$ have separators of size at most $t+1$; we use standard terminology for tree decompositions and treewidth, as in~\cite{cygan-book}.

Our conditional lower bounds use ETH and Gap-ETH. ETH asserts that $3$-SAT on $N$ variables cannot be solved in $2^{o(N)}$ time~\cite{impagliazzo2001complexity}. Gap-ETH asserts that for some constant $\delta>0$, no $2^{o(N)}$ time algorithm can distinguish satisfiable $3$-SAT instances from instances in which every assignment leaves at least a $\delta$ fraction of clauses unsatisfied~\cite{Dinur16,ManurangsiR17}.

\section{MSVC on Bounded-Degree Graphs}
\label{sec:bounded-degree-vc}

In this section, we analyze a simple approximation algorithm for bounded-degree graphs when a minimum vertex cover is available. This yields a polynomial-time algorithm on any graph class where a minimum vertex cover can be computed in polynomial-time. Examples are trees, bipartite graphs, graphs of bounded treewidth, and perfect graphs (which includes chordal, interval, comparability, and co-comparability graphs)~\cite{grotschel1984polynomial}.

Let $G=(V,E)$ be a graph with $m>0$ edges. A labeling $\varphi:V\to\{1,\ldots,|V|\}$ assigns every edge to the endpoint with smaller label. Thus the labeling induces a partition of $E$ into stars. If the number of edges in the nonempty stars are $s_1,\ldots,s_K$ and their centers receive labels $\ell_1,\ldots,\ell_K$, then the cost is $\sum_{i=1}^K \ell_i s_i$.
For a fixed multiset of star sizes, this expression is minimized by placing stars with more edges earlier. See \Cref{fig:example}.

Fix an integer $\Delta\ge2$, and suppose that $G$ has maximum degree at most $\Delta$. Let $C$ be a minimum vertex cover of $G$. We use the following
ordering algorithm.

\begin{center}
\noindent\fbox{\parbox{0.97\textwidth}{
\textbf{\textsc{Minimum-Cover Star Ordering.}}
Initially, every edge is uncovered. For $i{=}1,\ldots,|C|$, choose an
unprocessed vertex $v_i\in C$ of maximum uncovered degree, assign label $i$ to
$v_i$, and declare all currently uncovered edges incident to $v_i$ to be
covered. After all vertices of $C$ have been processed, assign the remaining
labels arbitrarily to $V\setminus C$.
}}
\end{center}

Note that this labeling is a bijection, and its inverse induces the ordering $\sigma_C$ as in the definition of MSVC. For $\Delta\ge2$, define
\begin{equation}\label{eq:RDelta}
    R_\Delta := \max_{1\le k\le \Delta} \frac{k(\Delta-1)}{\Delta-1+(k-1)^2}.
\end{equation}

\begin{theorem}\label{thm:min-cover-star-ordering} Let $G=(V,E)$ be a graph of maximum degree at most $\Delta\ge2$, and let $C$ be a minimum vertex cover of $G$. The \textsc{Minimum-Cover Star Ordering} $\sigma_C$ constructed from $C$ satisfies $\cost(\sigma_C)\le R_\Delta\,\mathrm{msvc}(G)$. Moreover, $R_\Delta\le (\sqrt{\Delta}+1)/2$.
\end{theorem}

\begin{figure}
    \centering
    \includegraphics[width=0.9\linewidth]{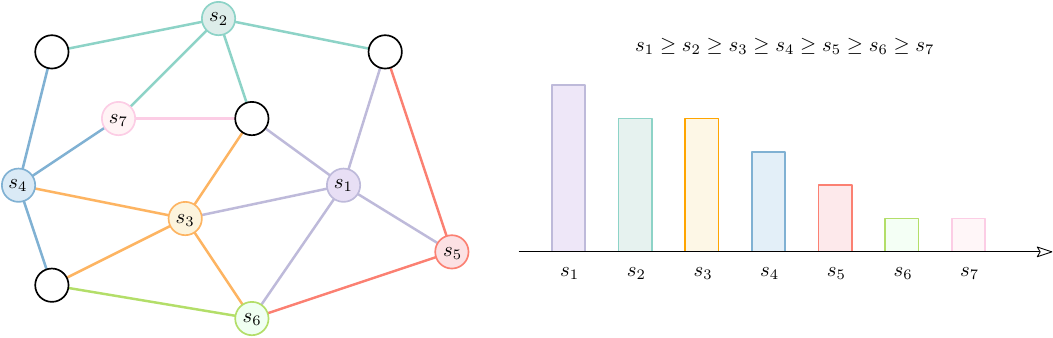}
    \caption{Illustration of the partition into stars: stars of larger size are assigned smaller labels. The white vertices are assigned arbitrarily large labels, as they do not contribute to the cost.}
    \label{fig:example}
\end{figure}

\begin{proof}
Let $M:=|C|=\tau(G)$. Each vertex in $C$ covers at least one distinct edge. Therefore, the algorithm produces exactly $M$ nonempty stars. If their sizes, in the order in which they are produced, are $s_1,\ldots,s_M$, then $\Delta\ge s_1\ge s_2\ge\cdots\ge s_M\ge1$, and $\sum_{i=1}^M s_i=m$.

\begin{restatable}{claim}{clmmincoveralgupper}\label{lem:min-cover-alg-upper}
Write $m=kM+r$, where $k=\lfloor m/M\rfloor$ and $0\le r<M$. Then
$1\le k\le\Delta$, and the cost satisfies $\cost(\sigma_C) \le k\frac{M(M+1)}{2}+\frac{r(r+1)}{2}$.
\end{restatable}

\begin{proof}
Note that $\cost(\sigma_C)=\sum_{i=1}^M i s_i$,
where $s_1\ge\cdots\ge s_M$ are integers in $[1,\Delta]$ with sum $m$. Among all such sequences, this weighted sum is maximized when the values are as balanced as possible. Indeed, if $s_i\ge s_{i+1}+2$, then replacing $(s_i,s_{i+1})$ by $(s_i-1,s_{i+1}+1)$ increases the weighted sum by $1$. Repeating this exchange yields $r$ entries equal to $k+1$ followed by $M-r$ entries equal to $k$. Therefore
\[
    \cost(\sigma_C)
    \le
    (k+1)\sum_{i=1}^{r}i
    +k\sum_{i=r+1}^{M}i
    =
    k\frac{M(M+1)}{2}+\frac{r(r+1)}{2}.
\]
\end{proof}

\begin{restatable}{claim}{lemmincoveroptlower}\label{lem:min-cover-opt-lower} As a lower bound, we have $\mathrm{msvc}(G) \ge \frac12\left(K^2+\frac{(m-K)^2}{\Delta-1}+m\right)$.
\end{restatable}

\begin{proof}
Let an optimal ordering induce $K$ nonempty stars, and order their sizes 
(i.e., number of edges) $a_1\ge\cdots\ge a_K$. Their centers form a vertex cover, so $K\ge M$. Also $K\ge m/\Delta$, and $M\ge m/\Delta$ because every star, and every vertex in a vertex cover, accounts for at most $\Delta$ edges.

Set $d:=\Delta-1$ and write $a_i=1+x_i$, where $0\le x_i\le d$ and $\sum_i x_i=m-K=:S$. The cost is $K(K+1)/2+\sum_i i x_i$. The second term is minimized by packing the extra units as far left as possible: if $i<j$, $x_i<d$, and $x_j>0$, moving one unit from $j$ to $i$ only decreases the cost. Thus, writing $S=qd+\rho$ with $0\le\rho<d$, we get
\[
    \sum_i i x_i
    \ge d\frac{q(q+1)}{2}+\rho(q+1)
    \ge \frac{S^2}{2d}+\frac S2,
\]
where the last inequality follows from
$dq(q+1)/2+\rho(q+1)-S^2/(2d)-S/2=\rho(d-\rho)/(2d)\ge0$. Hence any $K$-star partition has cost at least
\[
    f(K):=\frac12\left(K^2+\frac{(m-K)^2}{\Delta-1}+m\right).
\]
Finally, $f'(t)=(\Delta t-m)/(\Delta-1)$, so $f$ is nondecreasing for $t\ge m/\Delta$. Since $K\ge M\ge m/\Delta$, we have $\mathrm{msvc}(G)\ge f(K)\ge f(M)$.
\end{proof}

Put $d:=\Delta-1$ and use the division $m=kM+r$ from
Claim~\ref{lem:min-cover-alg-upper}. Let $\theta:=r/M\in[0,1)$. Then
$m=M(k+\theta)$ and $m-M=M(k-1+\theta)$. By Claim~\ref{lem:min-cover-alg-upper} and Claim~\ref{lem:min-cover-opt-lower}, we have 
\[
    2\cost(\sigma_C) \le M^2(k+\theta^2)+m \qquad\text{and}\qquad 2\,\mathrm{msvc}(G) \ge M^2(1+{(k-1+\theta)^2}/{d})+m.
\]
Therefore,
\[
    \frac{\cost(\sigma_C)}{\mathrm{msvc}(G)} \le \frac{k+\theta^2+m/M^2}{1+(k-1+\theta)^2/d+m/M^2} \le \frac{d(k+\theta^2)}{d+(k-1+\theta)^2} =: F_k(\theta).
\]
A direct derivative check shows that $F_k$ has no strict interior maximum on $[0,1]$. Hence the maximum is attained at an endpoint, and checking both endpoints proves that $F_k(\theta) \le R_\Delta$. Finally, $g(x)=dx/(d+(x-1)^2)$ satisfies $g'(x)=d(\Delta-x^2)/(d+(x-1)^2)^2$, so its  maximum is $g(\sqrt\Delta)=(\sqrt\Delta+1)/2$. Therefore $R_\Delta\le(\sqrt\Delta+1)/2$.
This completes the proof.
\end{proof}

For general graphs, the best approximation ratio that we use for comparison is the $16/9$-approximation of~\cite{BansalBFT21}. The following table gives the value of $R_\Delta$ for small values of $\Delta$, stopping at the first value that exceeds $16/9$. Thus the bound is better than $16/9$ for $\Delta\le6$.
\[
\renewcommand{\arraystretch}{1.3}
\begin{array}{c|@{\hspace{0.35cm}}ccccc}
\Delta     & 3 & 4 & 5 & 6 & 7 \\
\midrule
R_\Delta  & \frac{4}{3} & \frac{3}{2} & \frac{8}{5} & \frac{5}{3} & \frac{9}{5} \\
\end{array}
\]

\medskip\noindent\textbf{Tightness of Minimum-Cover Star Ordering}.
We now show that the approximation ratio in the \textsc{Minimum-Cover Star Ordering} algorithm cannot be improved when a minimum vertex cover is prescribed. Fix $\Delta\ge2$, put $d:=\Delta-1$, and choose $k\in\{1,\ldots,\Delta\}$ that maximizes $R_\Delta$.

For an integer $L\ge2$, 
let $\mathbb{Z}_L$ denote the integers modulo $L$. 
We will define a bipartite graph $G_L$ with vertex set $A \cup B$. 
Let $A=\{a_{\ell,t}:\ell\in\mathbb Z_L,\ 1\le t\le d\}$. The opposite side of the bipartition consists of high-degree and low-degree vertices:
\[
    B^{\mathrm H} = \{b^{\mathrm H}_{\ell,j}:\ell\in\mathbb Z_L,\ 1\le j\le k-1\},
    \qquad
    B^{\mathrm L} = \{b^{\mathrm L}_{\ell,t}:\ell\in\mathbb Z_L,\ k\le t\le d\}.
\]
Empty index ranges are allowed. Let $B=B^{\mathrm H}\cup B^{\mathrm L}$. Then $|A|=|B|=dL=:M$, and the number of high-degree vertices is $q:=|B^{\mathrm H}|=(k-1)L$. The edge set consists of three families. First, for every $\ell\in\mathbb Z_L$, every $1\le t\le d$, and every $1\le j\le k-1$, add $\{a_{\ell,t},b^{\mathrm H}_{\ell,j}\}$. Second, for every $\ell\in\mathbb Z_L$ and $1\le j\le k-1$, add $\{b^{\mathrm H}_{\ell,j},a_{\ell+1,j}\}$. Third, for every $\ell\in\mathbb Z_L$ and $k\le t\le d$, add $\{b^{\mathrm L}_{\ell,t},a_{\ell,t}\}$. Since $L\ge2$, the second family does not duplicate edges from the first family, so the graph is simple. Let the resulting graph be~$G_L$.

\begin{restatable}{proposition}{thmmincovertight}\label{thm:min-cover-tight}
The graph $G_L$ is bipartite, has maximum degree at most $\Delta$, and has $A$ as a minimum vertex cover. If the \textsc{Minimum-Cover Star Ordering} is applied to this specified vertex cover $A$, then $\lim_{L\to\infty} \frac{\cost(\sigma_A)}{\mathrm{msvc}(G_L)} =R_\Delta$.
\end{restatable}

\begin{proof}
The graph is bipartite with sides $A$ and $B$. From the construction, every $a_{\ell,t}$ has degree $k$, every high vertex has degree $d+1=\Delta$, and every low vertex has degree one. Hence $\Delta(G_L)\le\Delta$.

The second and third edge families form a perfect matching between $A$ and $B$, so every vertex cover has size at least $M$. Since both sides of the bipartition are vertex covers of size $M$, the set $A$ is a minimum vertex cover.

If the algorithm is applied to this prescribed cover $A$, then no edge is covered before its endpoint in $A$ is processed, because $A$ is independent. Thus every algorithmic star has size $k$, and
\begin{equation}\label{eq:tight-alg-cost}
    \cost(\sigma_A)=k\sum_{i=1}^M i=k\frac{M(M+1)}{2}.
\end{equation}
On the other hand, using $B$ as the set of star centers and placing the $q$ high-degree vertices first gives a feasible labeling of cost
\begin{equation}\label{eq:tight-B-cost}
    \mathcal W_B
    =\Delta\sum_{i=1}^{q}i+\sum_{i=q+1}^{M}i
    =\frac12\bigl(M(M+1)+d q(q+1)\bigr).
\end{equation}
Therefore $\mathrm{msvc}(G_L)\le\mathcal W_B$, and with $M=dL$ and $q=(k-1)L$, we get
\[
    \liminf_{L\to\infty}
    \frac{\cost(\sigma_A)}{\mathrm{msvc}(G_L)}
    \ge
    \frac{kd}{d+(k-1)^2}
    =R_\Delta.
\]
By \Cref{thm:min-cover-star-ordering}, since $A$ is a minimum vertex cover and $\Delta(G_L) \le \Delta$, we have that $\cost(\sigma_A)/\mathrm{msvc}(G_L)$ is at most $R_\Delta$ for every $L$. Combined with the preceding lower bound, the limit is exactly $R_\Delta$.
\end{proof}

\section{MSVC on \texorpdfstring{$d$}{d}-Regular Graphs} \label{sec:regular}

In this section we prove that MSVC admits a $1.184$-approximation on $d$-regular graphs. The algorithm is a modification of the regular-graph algorithm discussed by Feige et al.~\cite{FLT04} and Stankovi\'c~\cite{Stankovic25}. The main difference is that we run a Max-$k$-Vertex-Cover\footnote{Max-$k$-Vertex Cover is the problem of finding $k$ vertices that cover the largest number of edges in a graph.} algorithm for every integer value of $k$, rather than fixing $k=n/2$ or attempting to guess the value of the optimum. Moreover, our structure lemma (\Cref{lem:coverage}) gives a lower bound on the number of edges that an optimal ordering covers at every prefix.

Let $G=(V,E)$ be a $d$-regular graph, let $n=|V|$, and let $m=|E|=nd/2$. We assume that $d\ge 1$, since the problem is trivial when $E=\varnothing$. Following the notation of~Stankovi\'c~\cite{Stankovic25}, for an ordering $\sigma:\{1,2,\dots,n\}\to V$, define $\SVC(\sigma) := \frac{1}{m}\sum_{\{u,v\}\in E} \min\{\sigma^{-1}(u),\sigma^{-1}(v)\}$.
This is the average edge-cover time. Multiplying this quantity by $m$ gives the usual MSVC objective, so the normalization does not change approximation ratios. For $i=0,\ldots,n$, let $u_i^\sigma$ be the number of edges with both endpoints in $V\setminus \sigma(\{1,\dots,i\})$. Thus $u_0^\sigma=m$, 
$u_n^\sigma=0$, and
\begin{equation}\label{eq:svc-uncovered}
    \SVC(\sigma)=\frac{1}{m}\sum_{i=0}^{n-1}u_i^\sigma.
\end{equation}
Indeed, an edge covered at step $j$ is counted among the uncovered edges after steps $0,1,\ldots,j{-}1$, and is therefore counted exactly $j$ times in the right-hand side of~\eqref{eq:svc-uncovered}.

Let $\sigma^\star$ be an optimal ordering, write $u_i=u_i^{\sigma^\star}$, and let $\OPT=\SVC(\sigma^\star)$. Define $c_i=u_{i-1}-u_i$, the number of new edges covered at step $i$.
The following observation is a standard exchange argument~\cite{FLT04, Stankovic25}. See \Cref{fig:cover}(a).

\begin{restatable}{observation}{obsnonincreasing}\label{obs:nonincreasing}
For every optimal ordering $\sigma^\star$, the sequence $c_1,c_2,\ldots,c_n$ is
non-increasing.
\end{restatable}

\begin{proof}
Suppose that $c_i<c_{i+1}$ for some $i$. Let $x=\sigma^\star(i)$ and
$y=\sigma^\star(i+1)$, and let $a\in\{0,1\}$ indicate whether $\{x,y\}\in E$.
If $x$ and $y$ are swapped, then $y$ covers $c_{i+1}+a$ edges at step $i$,
whereas $x$ covers $c_i-a$ edges at step $i+1$. After step $i+1$ the same
vertices have been selected, so all later uncovered-edge counts are unchanged.
Since $c_{i+1}+a>c_i$, the swap strictly decreases $u_i$ and hence strictly
decreases the objective, contradicting optimality.
\end{proof}

We first record the range in which the optimum lies. Since a vertex covers at
most $d$ new edges, $u_i\ge m-id$ for every $i$. Summing the positive terms of
this inequality in~\eqref{eq:svc-uncovered} gives
$\OPT\ge n/4+1/2$. On the other hand, in a uniformly random ordering each
edge has expected cover time $(n+1)/3$, and hence $\OPT\le(n+1)/3$. Therefore
there is a value $\delta\in[0,1/12]$ such that
\begin{equation}\label{eq:opt-delta}
    \OPT=\frac{n}{4}+\frac{1}{2}+\delta n.
\end{equation}
The additive $1/2$ in~\eqref{eq:opt-delta} is useful for keeping the following
structure lemma exact for finite $n$. See \Cref{fig:cover}(a) for an illustration of the following structure lemma.

\begin{restatable}{lemma}{lemcoverage}\label{lem:coverage}
Let $G$ be a $d$-regular graph whose optimum is as in~\eqref{eq:opt-delta}.
Set $\tau=(1-\sqrt{\delta})/2$, let $r=\lceil \tau n\rceil$, and let
$s=u_r$ be the number of edges left uncovered by $\sigma^\star$ after step
$r$. Then $s\le \varphi\sqrt{\delta}\,m$, where
$\varphi=(1+\sqrt{5})/2$. Equivalently, the first $r$ vertices of the optimal
ordering cover at least $(1-\varphi\sqrt{\delta})m$ edges.
\end{restatable}

\begin{proof}
Let $q=1-2r/n$, and let $c=c_r$. Since $r\ge \tau n$, we have
$q\le 1-2\tau=\sqrt{\delta}$. By Observation~\ref{obs:nonincreasing}, for
every $i=0,\ldots,n$ we have
\begin{equation}\label{eq:ui}
    u_i\ge \max\bigl\{m-id,\ s+(r-i)c,\ 0\bigr\}.
\end{equation}
The first bound holds because at most $d$ edges can be covered at each step,
and the third bound is immediate. For the second bound, if $i\le r$, then
$u_i-u_r=\sum_{j=i+1}^{r}c_j\ge(r-i)c$, while if $i\ge r$, then
$u_r-u_i=\sum_{j=r+1}^{i}c_j\le(i-r)c$.

If $c=d$, then $c_j=d$ for every $j\le r$, so
$s=m-rd=mq\le\sqrt{\delta}\,m$. If $c=0$, then $c_j=0$ for every $j>r$, so
$s=u_r-u_n=0$. We may therefore assume that $0<c<d$.

\begin{claim}\label{clm:cost}
Under the preceding assumptions,
\[
    \OPT \ge \frac{n}{4}+\frac{1}{2} + \frac{(2s-ncq)^2}{4nc(d-c)}.
\]
\end{claim}

\begin{proof}
For real $x$, define $L_1(x)=m-dx$, $L_2(x)=s+(r-x)c$, and $L_3(x)=0$.
Let $a$ be the unique real number satisfying $L_1(a)=L_2(a)$, and let $b$ be
the unique real number satisfying $L_2(b)=0$. Since $0<c<d$, these values are
\[
    a=\frac{m-s-rc}{d-c}
    \qquad\text{and}\qquad
    b=r+\frac{s}{c}.
\]
We have $0\le a\le r\le b\le n$. Indeed,
$m-s=\sum_{j=1}^{r}c_j\ge rc$ gives $a\ge0$, while $m-s\le rd$ gives
$a\le r$. Moreover, $s=\sum_{j=r+1}^{n}c_j\le(n-r)c$ gives $b\le n$.
Consequently, the pointwise maximum in~\eqref{eq:ui} is $L_1$ before $a$,
$L_2$ between $a$ and $b$, and $L_3$ after $b$.

Set $p=\lfloor a\rfloor$, $\ell=\lfloor b\rfloor$,
$\alpha=a-p$, and $\gamma=b-\ell$. Since $u_n=0$, equations
~\eqref{eq:svc-uncovered} and~\eqref{eq:ui} imply
\begin{equation*}
    m\cdot\OPT =\sum_{i=0}^{n}u_i \ge
    \sum_{i=0}^{p}(m-id)
    +
    \sum_{i=p+1}^{\ell}\bigl(s+(r-i)c\bigr).
\end{equation*}
The two arithmetic sums are
\[
    \sum_{i=0}^{p}(m-id)=(p+1)m-\frac{d}{2}p(p+1)
\]
and
\[
    \sum_{i=p+1}^{\ell}\bigl(s+(r-i)c\bigr)
    =
    (\ell-p)(s+rc)
    -
    \frac{c}{2}\bigl(\ell(\ell+1)-p(p+1)\bigr).
\]
Let $M=L_1(a)=L_2(a)=m-da$. Since $b=a+M/c$ and $s+rc=M+ac$, substituting $p=a-\alpha$ and $\ell=b-\gamma$ into the preceding sums gives
\begin{align*}
    m\OPT
    \ge{}&
    am-\frac{da^2}{2}+\frac{(m-da)^2}{2c}+\frac{m}{2} \\
    &\quad+
    \frac{d-c}{2}\alpha(1-\alpha)
    +
    \frac{c}{2}\gamma(1-\gamma).
\end{align*}
The final two terms are non-negative. Writing $a=tn$ and using $m=nd/2$,
we obtain
\[
    \OPT
    \ge
    \frac{n}{4}+\frac{1}{2}
    +
    \frac{n(d-c)}{c}\left(\frac{1}{2}-t\right)^2.
\]
Finally, the equation $L_1(tn)=L_2(tn)$ gives
\[
    1-2t
    =
    \frac{2s-ncq}{n(d-c)}.
\]
Substituting this identity into the previous lower bound proves the claim.
\end{proof}

Combining Claim~\ref{clm:cost} with~\eqref{eq:opt-delta} gives
$(2s-ncq)^2\le 4\delta n^2c(d-c)$, and hence
$2s\le ncq+2n\sqrt{\delta c(d-c)}$. Set $x=c/d$. Since $2m=nd$ and
$q\le\sqrt{\delta}$, division by $nd$ yields
\[
    \frac{s}{m}
    \le
    qx+2\sqrt{\delta x(1-x)}
    \le
    \sqrt{\delta}\bigl(x+2\sqrt{x(1-x)}\bigr).
\]
The function $x+2\sqrt{x(1-x)}$ is maximized on $[0,1]$ at
$x=(1+\sqrt{5})/(2\sqrt{5})$, where its value is $\varphi$. Therefore
$s\le\varphi\sqrt{\delta}\,m$.
\end{proof}

\begin{figure}
    \centering
    \includegraphics[width=0.95\linewidth]{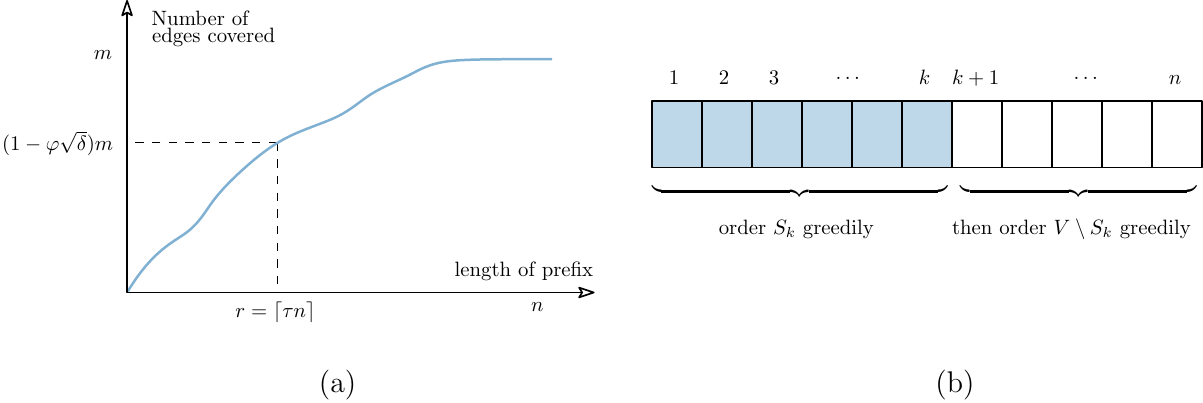}
    \caption{(a) Illustration of the number of edges covered by an optimal solution in a given prefix. Note that in an optimal solution the number of newly covered edges should decrease as the prefix grows. (b) Step~(b) of the \textsc{MSVC-Regular} algorithm.}
    \label{fig:cover}
\end{figure}

\noindent\textbf{The algorithm.}
We use the following approximation guarantee for Max-$k$-Vertex-Cover.

\begin{theorem}[Raghavendra and Tan~\cite{RT12}, Austrin and Stankovi\'c~\cite{AS19}]\label{thm:max-k-vc}
There is a polynomial-time algorithm which, given an edge-weighted graph and an integer $k$, outputs a set $S\subseteq V$ of size $k$ whose covered edge weight is at least $\beta$ times the optimum, where $\beta\ge0.929$.
\end{theorem}

\noindent\emph{Remark.}
Max-$k$-Vertex-Cover is the monotone special case of CC-Max-2SAT. Raghavendra and Tan \cite{RT12} state the rounded guarantee $0.92$; a more detailed analysis of the same algorithm gives the value approximately $0.929$ \cite{AS19} (See also~\cite{JainKPSSSU26}). We use the same worst-case approximation factor $\beta$ for all cardinalities $k$. This is deliberate, although the balanced case $k=n/2$ admits a better guarantee and was the split used in the earlier regular-graph analysis~\cite{Stankovic25}. Our proof needs the set size $r=\lceil(1-\sqrt{\delta})n/2\rceil$, which may be smaller than $n/2$ and is not known to the algorithm. The improvement below comes from trying all values of
$k$ and then applying Lemma~\ref{lem:coverage}, rather than from improving the Max-$k$-Vertex-Cover~subroutine.

\medskip
Fix the constant $n_0=6000$. The exhaustive-search step below is included only to remove finite-size rounding terms from the final approximation ratio. Since $n_0$ is a fixed constant, this does not affect polynomial-time solvability.

\begin{center}
\noindent\fbox{
\begin{minipage}{0.97\textwidth}
\textbf{\textsc{MSVC-Regular Algorithm}$(G)$.}
\begin{enumerate}[label=\arabic*.,leftmargin=2em,itemsep=0.25em,topsep=0.3em]
    \item If $n<n_0$, enumerate all vertex orderings and return an optimal one.
    \item Run the greedy MSVC algorithm and let $\sigma_{\mathrm{greedy}}$ be
    its output.
    \item For every integer $k\in\{1,2,\ldots,n\}$:
    \begin{enumerate}[label=\alph*.,leftmargin=1.5em,itemsep=0.15em,topsep=0.2em]
        \item Run the $\beta$-approximate Max-$k$-Vertex-Cover algorithm on
        $G$ and obtain a set $S_k\subseteq V$ with $|S_k|=k$.
        \item Construct $\sigma_k$ by derandomizing, via conditional expectation, a uniformly random permutation of $S_k$ followed by one of $V\setminus S_k$; see Lemma~\ref{prop:two-phase} and \Cref{fig:cover}(b).
    \end{enumerate}
    \item Return the minimum-cost ordering among $\sigma_{\mathrm{greedy}}$
    and $\{\sigma_k:k\in[n]\}$.
\end{enumerate}
\end{minipage}}
\end{center}

Here $\sigma_{\mathrm{greedy}}$ is the ordering produced by the standard greedy rule for MSVC: beginning from the empty prefix, repeatedly append a vertex incident to the largest number of still uncovered edges, breaking ties arbitrarily, and delete the edges it covers. The analysis of the following proposition adapts a similar argument from the analysis of the greedy algorithm in~\cite{FLT04} to the setting of step~(b) of the algorithm described above.

\begin{restatable}{lemma}{proptwophase}\label{prop:two-phase}
Let $S\subseteq V$ have size $k$, and suppose that $S$ covers a fraction $w$ of the edges. There is a polynomial-time computable ordering $\sigma_S$ that places $S$ first and $V\setminus S$ second and satisfies $\SVC(\sigma_S) \le
 (n+2k+1-k(k+1)/n-(n-1)w)/3$.
\end{restatable}

\begin{proof}
Let $a,b,c$ be the numbers of edges internal to $S$, crossing the cut, and
internal to $V\setminus S$. Since $G$ is $d$-regular, the degree sum over $S$
gives $2a+b=dk$, while $a+b=wm$ and $m=nd/2$. Hence
\begin{equation}\label{A.1}
 \frac am=\frac{2k}{n}-w,\qquad
 \frac bm=2w-\frac{2k}{n},\qquad
 \frac cm=1-w.
\end{equation}
Choose independently a uniformly random permutation of $S$ and one of
$V\setminus S$, and concatenate them. If an edge has two endpoints in a
block of size $q$, the expected minimum of their two distinct positions is
\[
 \frac{1}{\binom q2}\sum_{1\le i<j\le q}i=\frac{q+1}{3}.
\]
Thus an edge internal to $S$ has expected cover time $(k+1)/3$. A crossing
edge is covered when its endpoint in $S$ appears, whose expected position is
$(k+1)/2$. Finally, an edge internal to $V\setminus S$ waits through the first
$k$ positions and then has expected additional cover time $(n-k+1)/3$.
By (\ref{A.1}) and linearity of expectation,
\begin{align*}
\mathbb E[\SVC]
&=\left(\frac{2k}{n}-w\right)\frac{k+1}{3}
 +\left(2w-\frac{2k}{n}\right)\frac{k+1}{2}
 +(1-w)\left(k+\frac{n-k+1}{3}\right)\\
&=\frac{n+2k+1-k(k+1)/n-(n-1)w}{3}.
\end{align*}

It remains to derandomize. Suppose a prefix of the two permutations has
already been fixed. For every eligible remaining vertex $v$, let $Q(v)$ be
the expected final cost conditional on placing $v$ next and completing all
unfixed positions uniformly at random. The current conditional expectation
is the average of $Q(v)$ over the eligible vertices. Therefore some $v$ has
$Q(v)$ no larger than the current conditional expectation; fixing such a
vertex preserves the bound. Repeating this for all positions yields a
deterministic ordering whose cost is at most the displayed expectation.
Each $Q(v)$ is computable in polynomial time by summing edge contributions:
contributions of already covered edges are fixed, while each remaining edge
has an expectation determined by whether its endpoints are fixed, in the same
block, or in different blocks. Hence the derandomization is polynomial-time.
\end{proof}

\begin{theorem}\label{thm:regular}
There is a polynomial-time $1.184$-approximation algorithm for \textsc{Minimum Sum Vertex Cover} on $d$-regular graphs.
\end{theorem}

\begin{proof}
Let $G$ have $n\ge n_0$ vertices, and write its optimum as in~\eqref{eq:opt-delta}. Set $\tau=(1-\sqrt{\delta})/2$ and $r=\lceil\tau n\rceil$. The algorithm does not know $\delta$, but it tries every integer value of $k$, so it considers $k=r$.

By Lemma~\ref{lem:coverage}, the first $r$ vertices of an optimal ordering
cover at least $(1-\varphi\sqrt{\delta})m$ edges. Hence the set $S_r$
returned by Theorem~\ref{thm:max-k-vc} covers a fraction at least
$w_\delta=\beta(1-\varphi\sqrt{\delta})$ of the edges. Since $\tau n\le r<\tau n+1$, we have
$2r\le2\tau n+2$ and $r(r+1)/n\ge\tau^2n+\tau$. Applying
Lemma~\ref{prop:two-phase} with $k=r$ and using
$r\le(1-\sqrt{\delta})n/2+1$ gives

\[
 \SVC(\sigma_r)\le
 \frac{n+2r+1-r(r+1)/n-(n-1)w}{3}
 \le\frac{n(1+2\tau-\tau^2-w)+3-\tau+w}{3}.
\]
Substituting
$\tau=(1-y)/2$ and $w=\beta(1-\varphi y)$ into the leading coefficient yields
\begin{align*}
 F_\beta(y)&:=\frac{1+2\tau-\tau^2-w}{3}
 =\frac{1+(1-y)-(1-y)^2/4-\beta(1-\varphi y)}{3}\\
 &=\frac{7-4\beta+(4\beta\varphi-2)y-y^2}{12}.
\end{align*}
Since $0\le\tau,w\le1$, we have $\SVC(\sigma_r) \le nF_\beta(y)+4/3$. From $\OPT=n(1/4+\delta)+1/2$, we obtain
\[
 \frac{\SVC(\sigma_r)}{\OPT}
 \le
 \frac{7-4\beta+(4\beta\varphi-2)y-y^2}{3+12y^2}
 +\frac{16}{3n}.
\]
For $\beta=0.929$ and $0\le y\le1/\sqrt{12}$, the first term is at
most $1.1830$ (its maximum is attained near $y=0.1321$). For $n\ge6000$,
the finite-size term is at most $0.0009$, so the ratio is below $1.184$.
For $n<n_0$ the algorithm returns an optimal ordering by exhaustive search. Thus the algorithm is a polynomial-time $1.184$-approximation for all $d$-regular graphs.
\end{proof}

\section{An Exact Algorithm Parameterized by Vertex Cover}
\label{sec:vc-exact}

In this section we give an exact algorithm for MSVC when a vertex cover of size $k$ is given. The algorithm is close in spirit to the approaches of \cite{AP25} and \cite{LCGW25}. Its main viewpoint is that an ordering can be described by its right degrees (to be defined below). In an optimal ordering these right degrees may be assumed to be non-increasing, and consecutive vertices with the same right degree can be treated as one block. Thus, instead of guessing the complete order of all vertices, we guess the block structure and decide which vertices can be placed in each block. The vertex cover makes this search space small, because all vertices outside the cover form an independent set and can be grouped by their neighbourhoods in the cover.

Let $G=(V,E)$ be a graph with $|V|=n$ and $|E|=m$. Let $X\subseteq V$ be a vertex cover with $|X|=k$, and put $I:=V\setminus X$. Thus $I$ is an independent set. For an ordering $\sigma:\{1,2,\ldots,n\}\to V$, the \emph{right degree} of a vertex $v$ is $\rd_\sigma(v):=|\{u\in N(v):\sigma^{-1}(v)<\sigma^{-1}(u)\}|$. We drop the subscript when the ordering is clear. The MSVC cost of $\sigma$ can then be written as $\sum_{v\in V}\sigma^{-1}(v)\rd_\sigma(v)$.

\medskip\noindent\textbf{Right-degree blocks.}
For an ordering $\sigma$, define its right-degree sequence as $r^\sigma=(r^\sigma_1,\ldots,r^\sigma_n)$, where $r^\sigma_i:=\rd_\sigma(\sigma(i))$. By Observation~\ref{obs:nonincreasing}, there is an optimal ordering whose right-degree sequence is non-increasing. We therefore describe such an ordering by grouping consecutive equal entries of this sequence.

Consider distinct values $d_1>d_2>\cdots>d_q\ge 0$ and positive integers $s_1,\ldots,s_q$. The intended meaning is that the first $s_1$ vertices in $r^\sigma$ have right degree $d_1$, the next $s_2$ vertices in~$r^\sigma$ have right degree $d_2$, and so on. For an $n$-vertex, $m$-edge graph, any feasible choice~must satisfy $\sum_{j=1}^q s_j=n$ and $\sum_{j=1}^q s_jd_j=m$, because every vertex gets one right degree value and every edge contributes exactly once to the total right degree. Note that these equalities are necessary, but not sufficient. The vertices must still be assignable to the corresponding~blocks.

Set $p_0:=0$ and $p_j:=s_1+\cdots+s_j$ for $j\in\{1,\ldots,q\}$. The $j$th block is the interval of positions $p_{j-1}+1,\ldots,p_j$, and every vertex placed in this block is required to have right degree $d_j$. Thus an ordering realizes the chosen values and block sizes if $\rd_\sigma(\sigma(i))=d_j$ for every $i\in\{p_{j-1}+1,\ldots,p_j\}$ and every $j$. Equivalently, the blocks are exactly the equal-right-degree classes of the ordering, listed from larger to smaller right degree.

Once the values $d_j$ and sizes $s_j$ are fixed, the contribution of every position is fixed. Hence all orderings realizing these blocks have the same cost:
\begin{equation}\label{eq:block-cost}
   \mathrm{cost}(d,s)
   :=
   \sum_{j=1}^q d_j\sum_{i=p_{j-1}+1}^{p_j} i
   =
   \sum_{j=1}^q d_j\frac{(p_{j-1}+1+p_j)s_j}{2} .
\end{equation}
Thus, after fixing the right-degree values and block sizes, the algorithm only has to decide whether the vertices of $G$ can be assigned to the blocks so that the prescribed right degrees are obtained.

The next observation shows that the blocks of an optimal ordering may be treated as independent sets.

\begin{restatable}{observation}{obsindependent}\label{obs:independent}
In every optimal ordering, the two endpoints of any edge have different right degrees. In particular, each block of an optimal ordering is an independent set.
\end{restatable}

\begin{proof}
First suppose that $x$ and $y$ are consecutive vertices in an optimal ordering, with $x$ immediately before $y$, and that $xy\in E$ and $\rd(x)=\rd(y)$. If we exchange $x$ and $y$, then the only right-degree changes come from the edge $xy$: the right degree of $x$ decreases by one, and the right degree of $y$ increases by one. All other right degrees stay the same. Since $x$ moves one position later and $y$ moves one position earlier, the total cost decreases by one. This contradicts optimality.

Now suppose, for contradiction, that some optimal ordering contains an edge $uv$ whose endpoints have the same right degree. Among all such choices, take an optimal ordering and such an edge with minimum distance between its endpoints. By Observation~\ref{obs:nonincreasing}, all vertices between $u$ and $v$ have the same right degree as $u$ and $v$. If $u$ and $v$ are consecutive, the previous paragraph gives a contradiction. Otherwise, let $w$ be the vertex immediately before $v$. If $wv\in E$, then $wv$ is an equal-right-degree edge with smaller distance, contradicting the choice of $uv$. If $wv\notin E$, then exchanging $w$ and $v$ changes no right degree and leaves the cost unchanged, while decreasing the distance between $u$ and $v$. This again contradicts the choice of $uv$. Therefore no such edge exists.
\end{proof}

This independence property is what makes the block viewpoint useful. If no edge joins two vertices of the same block, then the later neighbours of a vertex are exactly its neighbours in strictly later blocks. Therefore the right degree of a vertex depends only on the block assigned to each vertex, and not on the internal order of vertices inside a block. In Section~\ref{sec:separator} we state this as an order-free feasibility test. In the present section, we exploit the given vertex cover to carry out this feasibility test in FPT time.

The next observation is the reason the vertex-cover parameter gives a small search space.
 
\begin{restatable}{observation}{obsvcfewblocks}\label{obs:vc-few-blocks}
Let $X$ be a vertex cover of size $k$. In any optimal ordering whose right degrees are non-increasing and whose equal right degree classes are independent, there are at most $2k+1$ non-empty blocks.
\end{restatable}

\begin{proof}
Every vertex of $I$ has all its neighbours in $X$, and therefore has right degree at most $k$. Hence every block of right degree larger than $k$ contains only vertices of $X$. There are at most $k$ such high-degree blocks. All remaining blocks have right degree in $\{0,1,\ldots,k\}$, so there are at most $k+1$ of them. See \Cref{fig:blocks}.
\end{proof}

\begin{figure}
    \centering
    \includegraphics[width=0.85\linewidth]{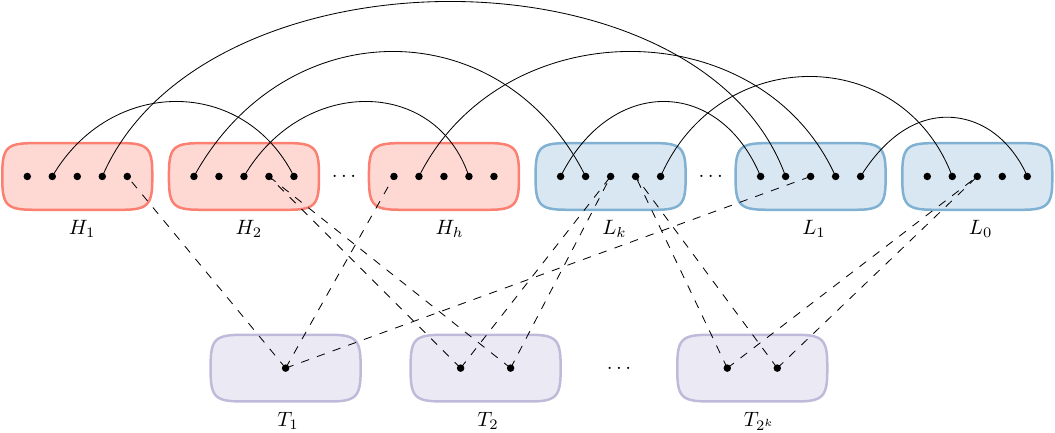}
    \caption{Skeleton obtained from the vertex cover $X$: The red color represents the block with high right-degree, the blue color represents the low right-degree block and the purple color represents the types of vertices in $I$.}
    \label{fig:blocks}
\end{figure}

\medskip\noindent\textbf{Skeletons and placement of the cover.}
We enumerate the possible placements of the vertices of $X$ into at most $2k+1$ blocks. A \emph{skeleton} consists of an integer $0\le h\le k$ and a map $\alpha:X\to \{H_1,\ldots,H_h,L_k,L_{k-1},\ldots,L_0\}$, where the intended order of the blocks is
\[
   H_1\prec H_2\prec\cdots\prec H_h\prec L_k\prec L_{k-1}\prec\cdots\prec L_0 .
\]
The block $L_r$ is intended to hold the vertices of right degree exactly $r$, and the blocks $H_i$ those of right degree larger than $k$. We stress the orientation of the low blocks: \emph{larger} right degree comes \emph{earlier}, so for a vertex in $L_r$ its later neighbours are precisely its neighbours in the strictly later low blocks $L_{r-1},\ldots,L_0$. Empty low blocks are allowed and ignored in the final ordering, and empty high blocks are discarded.
 
The number of skeletons is at most $\sum_{h=0}^{k}(h+k+1)^k = 2^{O(k\log k)}$. For a skeleton, we first check that every block is independent on $X$, i.e.\ no edge of $G[X]$ has both endpoints assigned to the same block. For a high block
$H_i$, the right degree of a vertex $x\in X\cap H_i$ is forced by the skeleton,
\[
   D_i(x):=|N_I(x)|+|\{y\in N_X(x): \alpha(y)\succ H_i\}| ,
\]
its independent neighbours (all of which lie in later low blocks) together with its cover neighbours placed in later blocks. We reject the skeleton unless all
vertices in $X\cap H_i$ share one value, denoted $D_i$, and $D_1>D_2>\cdots>D_h>k$.

It remains to assign the vertices of $I$ to the low blocks $L_k,\ldots,L_0$. For the low blocks, write
\[
   X_r := X\cap L_r,\qquad
   X_{<r}:=\bigcup_{r'<r}X_{r'},\qquad
   X_{>r}:=\bigcup_{r'>r}X_{r'},\qquad
   X_{\mathrm{low}}:=\bigcup_{r=0}^{k}X_r ,
\]
so that, by the orientation above, $X_{<r}$ are the cover vertices in low blocks \emph{after} $L_r$ and $X_{>r}$ those in low blocks \emph{before} $L_r$.
 
\medskip\noindent\textbf{Types outside the vertex cover.} Partition the vertices of the independent set $I$ based on their neighbourhoods in the cover~$X$. There are at most $2^k$ neighbourhoods, we call each one a \emph{type}. 
That is, each $T\subseteq X$ represents a type. Let $I_T:= \{u\in I : N(u)=T\}$. We say that the vertices of $I_T$ have type $T$. Let $\mu_T := |I_T|$. We refer to $\mu_T$ as the multiplicity of~$T$.

Since these multiplicities are fixed across all skeletons, we store them once as preprocessing~in a table: allocate an array indexed by the $k$-bit mask of $T\subseteq X$, initialize its $2^k$ entries to~$0$, and, for each $u\in I$, form the mask of $N(u)$ by scanning the edges at $u$ and increment the corresponding entry. The masks are built in $\sum_{u\in I}\deg(u)=O(m)$ time, because $I$~is independent and so every edge has exactly one endpoint in $I$, together with the $O(2^k)$ initialization this is $O(2^k+n+m)$ in total. The $2^k$ term is dominated by the $2^{O(k\log k)}$ spent below and so is absorbed into the running time, while the dependence on the graph size stays linear in $n+m$.\footnote{If one prefers to keep even the preprocessing within $O(n+m)$, at most $n$ masks that occur with non-zero multiplicities can instead be radix-sorted and grouped, avoiding the $2^k$ array altogether.}

Fix a skeleton. A vertex of independent set $I$ of type $T$ may be placed in the low block $L_r$ precisely when
\begin{equation}\label{eq:admissible}
   T\cap X_r=\varnothing
   \qquad\text{and}\qquad
   |T\cap X_{<r}|=r .
\end{equation}
The first condition keeps $L_r$ independent and the second makes the vertex have exactly $r$ later neighbours (all of which lie in $X_{<r}$ and right degree is equal to $r$). We call such a block $L_r$ \emph{admissible} for $T$. There are at most $k+1$ blocks admissible for $T$.
 
Now fix a cover vertex $x\in X_r$. Its neighbours in later $X$-blocks contribute $a_x := |N_X(x)\cap X_{<r}|$ to its right degree, so to reach right degree $r$ the vertices of $I$ placed after $x$ must contribute exactly $b_x := r-a_x$.
If $b_x<0$ for some $x$, we reject the skeleton. Otherwise $0\le b_x\le k$, and the total remaining demand is $B:=\sum_{x\in X_{\mathrm{low}}} b_x \le k^2$.
 
For a placement of one vertex of type $T$ into an admissible block $L_r$, define its \emph{demand-contribution vector} $p(T,r)\in\{0,1\}^{X_{\mathrm{low}}}$ by
\[
   p(T,r)_x =
   \begin{cases}
      1, & \text{if } x\in T \text{ and } \alpha(x)=L_s \text{ for some } s>r,\\
      0, & \text{otherwise.}
   \end{cases}
\]
that is, $x$ is a cover-neighbour of the type that is lying in a block \emph{before} $L_r$, so the placed vertex in $L_r$ is one of $x$'s later neighbours and supplies one unit of the demand towards $b_x$. A placement is \emph{neutral} if $p(T,r)=(0,\dots,0)$, i.e. it supplies no demand at all. 

For each type $T$ there is at most one neutral admissible block, which we denote by $L_{z(T)}$ (if it exists). Indeed $p(T,r)=0$ means $T\cap X_{>r}=\varnothing$, so every low-block cover-neighbour of $T$ lies in $X_{\le r}$, combined with admissibility $T\cap X_r=\varnothing$ and $|T\cap X_{<r}|=r$, all $|T\cap X_{\mathrm{low}}|$ of them lie in $X_{<r}$. Hence $z(T)=|T\cap X_{\mathrm{low}}|$ is the neutral block if and only if that block is admissible for $T$.

\medskip\noindent\textbf{Default placements and exceptions.}
Intuitively, the neutral block of a type is the cheapest to use, because filling it supplies no contributions towards demands. We therefore fix a \emph{default placement}: every vertex of a type $T$, that has a neutral block, is put in $L_{z(T)}$, and types with no neutral block are left unplaced. Under the default, no demand is supplied, so the demands $b_x$ are met only after we \emph{deviate} from it by moving some vertices into non-neutral admissible blocks. Two observations bound how much deviation is needed:

\emph{Each deviation supplies at least one unit of demand.} Moving a vertex of type $T$ from its neutral block to a non-neutral admissible block $L_r$ adds the non-zero vector $p(T,r)$ to the supplied demand, hence at least one unit. placing a vertex of a type with no neutral block likewise adds a non-zero vector $p(T,r)$.
 
\emph{At most $B$ vertices deviate.} The total demand to be supplied is exactly $B=\sum_x b_x\le k^2$, and each supplied unit comes from one deviating vertex, so at most $B$ vertices are placed non-defaultly. In particular, if a type with no neutral block has multiplicity exceeding $B$, the skeleton is infeasible. Moreover, from any type at most $\min\{\mu_T,B\}$ vertices ever leave the neutral block.
 
Consequently we need not track the placement of each independent vertex individually. It suffices to record, on top of the default, the few deviations and the demand they have supplied so far. Let $c^0_r:=\sum_{T:\,z(T)=r}\mu_T$ be the number of independent vertices placed in $L_r$ by the default.

\medskip\noindent\textbf{The dynamic program.}
Fix an enumeration $T_1,T_2,\ldots,T_{2^k}$ of the types. The dynamic program processes the types in this order, and a state $(i,\,j,\,\beta,\,\delta)$ records that types $T_1,\ldots,T_{i-1}$ have been fully placed and that $j$ copies of the current type $T_i$ have been placed so far, where
\[
   1\le i\le 2^k,\qquad 0\le j\le\min\{\mu_{T_i},B\},\qquad
   \beta\in\!\!\prod_{x\in X_{\mathrm{low}}}\!\!\{0,\ldots,b_x\}, \qquad
   \delta\in\{-B,\ldots,B\}^{k+1} .
\]
Here $\beta$ is the demand supplied so far to each low cover vertex\footnote{Here we assume a fixed ordering on the vertices of $X_\mathrm{low}$ such that the result of the above product is a vector where each entry corresponds to a distinct vertex in $X_\mathrm{low}$.} 
and $\delta_r$ for $0 \le r \le k$, is the deviation of the current number of independent vertices in $L_r$ from the default count $c^0_r$. The parameter $j$ lets the program keep track of how many copies of the current type it has already placed. Since at most $B$ vertices deviate in total, $\beta\le b$ (component-wise) and $|\delta_r|\le B$ hold throughout, so the number of states is at most
\[
   2^k\,(B{+}1)\!\!\prod_{x\in X_{\mathrm{low}}}\!\!(b_x{+}1)\,(2B{+}1)^{k+1}
   \le (2^k)(k^2{+}1)\,(k{+}1)^k\,(2k^2{+}1)^{k+1}=2^{O(k\log k)} .
\]
From a state $(i,j,\beta,\delta)$ there are two kinds of transition.

\emph{1. Place one more copy of $T_i$} (allowed while $j<\min\{\mu_{T_i},B\}$): choose one of the at most $k+1$ admissible non-neutral blocks $L_r$ for $T_i$ and move to
state $(i,j{+}1,\beta',\delta')$, where
\[
   \beta'=\beta+p(T_i,r),\qquad
   \delta'=
   \begin{cases}
      \delta-e_{z(T_i)}+e_r, & \text{if $T_i$ has a neutral block $L_{z(T_i)}$,}\\[2pt]
      \delta+e_r, & \text{if $T_i$ has no neutral block,}
   \end{cases}
\]
and the transition is rejected if $\beta'\not\le b$ in some coordinate. In the first case the copy is \emph{displaced} from its default block $L_{z(T_i)}$ (one fewer there, one more in $L_r$); in the second the copy has no default and is simply placed in $L_r$.
 
\emph{2. Advance to the next type}: move to state $(i{+}1,0,\beta,\delta)$. If $T_i$ has a neutral block this is allowed for every value of $j$, and the $\mu_{T_i}-j$ copies not placed explicitly stay in $L_{z(T_i)}$, where the default already counts them. If $T_i$ has no neutral block it is allowed only when $j=\mu_{T_i}$, which forces every copy to be placed; in particular the skeleton is infeasible if $\mu_{T_i}>B$, or if some copy has no admissible block at all.
 
The program starts in state $(1,0,\mathbf 0,\mathbf 0)$. After the last type is closed it reaches the states $(2^k{+}1,0,\beta,\delta)$, of which the \emph{accepting} ones are those with $\beta=b$, i.e., every demand is met exactly. It is clear that each transition can be done in time which is polynomial in $k$, so the total  time to populate the dynamic program table is $2^{O(k\log{k})}$.
 
\medskip\noindent\textbf{Computing the cost.}
An accepting state $(2^k{+}1,0,b,\delta)$ determines the whole block structure. The high blocks have the forced sizes $|X\cap H_i|$ and right degrees $D_1,\ldots,D_h$, while the low block $L_r$ has right degree $r$ and final size $|X_r|+c^0_r+\delta_r$.
Listing the blocks in the skeleton order $H_1,\ldots,H_h,L_k,\ldots,L_0$ fixes the positions, and the cost follows from the block-cost formula~\eqref{eq:block-cost}. We return the cheapest accepting state over all skeletons.

\begin{center}
\noindent\fbox{\parbox{0.97\textwidth}{
\textbf{Algorithm \textsc{VC-MSVC}$(G,X)$.}
\begin{enumerate}[label=\arabic*., leftmargin=2em, itemsep=0.25em, topsep=0.25em]
   \item Compute the nonzero type multiplicities $\mu_T$, $T\subseteq X$.
   \item Enumerate all skeletons $\alpha:X\to\{H_1,\ldots,H_h,L_k,\ldots,L_0\}$, for all $0\le h\le k$.
   \item Reject skeletons whose blocks are not independent on $X$, or whose high blocks do not have well-defined strictly decreasing right degrees larger than $k$.
   \item For each surviving skeleton, run the above dynamic programming algorithm.
   \item Return an ordering of minimum cost over all accepting states of all skeletons, obtained by listing the blocks in decreasing right degree and placing the vertices of each block arbitrarily inside it.
\end{enumerate}
}}
\end{center}
 
\begin{restatable}{theorem}{thmmsvcvcexact}\label{thm:msvc-vc-exact}
Let $G$ be a graph and let $X$ be a vertex cover of size $k$. \textsc{Minimum Sum Vertex Cover} on $G$ can be solved exactly in time
$2^{O(k\log k)}+O(n+m)$.
\end{restatable}

\begin{proof} The multiplicities $\mu_T$ are computed once in $2^k+O(n+m)$ time, as shown above. Everything that follows is independent of $n$ and $m$ beyond this term. There are $2^{O(k\log k)}$ skeletons, and for a fixed skeleton the independence and high-block checks take time polynomial in $k$ once the values $|N_I(x)|$ are known.
 
For one skeleton, the dynamic program has $2^{O(k\log k)}$ states: the type index and the per-type counter contribute factors $2^k{+}1$ and $B{+}1\le k^2{+}1$, while the target demand $B\le k^2$ gives at most $(k+1)^k$ demand vectors and $(2B+1)^{k+1}\le(2k^2+1)^{k+1}$ deviation vectors. Each state has at most $k+2$ transitions, and each is evaluated in $O(k)$ time. Hence one skeleton is handled in $2^{O(k\log k)}$ time, and multiplying by the number of skeletons keeps the total at $2^{O(k\log k)}$.
 
For correctness, every accepting state yields, by construction, a block assignment in which each block is independent and each vertex has exactly the right degree of its block, so listing the blocks in decreasing right degree produces a valid ordering of the stated cost. Conversely, fix an optimal ordering; by Observations~\ref{obs:nonincreasing} and~\ref{obs:independent} it has non-increasing right degrees and independent blocks, so by Observation~\ref{obs:vc-few-blocks} its blocks form one of the enumerated skeletons on $X$, and assigning each independent vertex to its block is an admissible placement meeting every demand exactly. This is reached by the dynamic program as an accepting state of that skeleton. Thus the minimum over all accepting states return the correct output.
\end{proof}

If the vertex cover $X$ is not given, we first compute a minimum vertex cover. For example, standard vertex-cover kernelization followed by branching computes such a set in $2^{O(k)}+O(n+m)$ time, where $k=\tau(G)$~\cite{cygan-book}. This is absorbed by the bound above. Thus MSVC parameterized by the vertex cover number admits an exact algorithm running in $2^{O(k\log k)}+O(n+m)$.

\section{An Exact Algorithm via Balanced Separators}
\label{sec:separator}

In this section we study the exact complexity of \textsc{Minimum Sum Vertex Cover}. For general graphs, a Held and Karp~\cite{HeldK61} style subset dynamic program runs in $\tilde{O}(2^n)$ time and one can show that, by the padded matching construction of \cite[Theorem~14]{FLT04}, through a linear-size reduction from bounded-degree gap Vertex Cover, this running time is qualitatively optimal: under Gap-ETH, MSVC admits no $2^{o(n)}$ time algorithm. For completeness, we state this in the following theorem and provide the proof in \Cref{sec:appendix-general}.

\begin{theorem}\label{thm:exact-general-graphs}
Let $G$ be an $n$-vertex graph. \textsc{Minimum Sum Vertex Cover} on $G$ can be solved in $2^n n^2$ time and $2^n n$ space. Moreover, assuming Gap-ETH, there is no algorithm running in $2^{o(n)}$ time.
\end{theorem}

We next show that the running time can be improved on hereditary graph classes that admit small balanced separators. Similar to \Cref{sec:vc-exact}, the algorithm has two ingredients. First, we enumerate the possible choices of right-degree values and block sizes for an optimal ordering. Second, for each such choice, we test whether the vertices of the graph can be assigned to the corresponding blocks by a dynamic program over a recursive separator decomposition.

Fix distinct values $d_1>d_2>\cdots>d_q\ge 0$ and positive integers $s_1,\ldots,s_q$ satisfying $\sum_j s_j=n$ and $\sum_j s_jd_j=m$. This data specifies $q$ blocks: block $j$ has size $s_j$ and target right degree~$d_j$. Recall that by Observation~\ref{obs:nonincreasing}, there is an optimal ordering whose right-degree sequence is non-increasing. By Observation~\ref{obs:independent}, the vertices with the same right degree form an independent set. Thus, once the values $d_j$ and sizes $s_j$ are fixed, the remaining question is not how to order vertices inside a block, but only whether the vertices can be assigned to the blocks so that the prescribed right degrees are respected.

We make this feasibility question explicit. Let $b:V\to\{1,\ldots,q\}$ be a map that assigns each vertex to a block. We say that $b$ realizes the chosen values and sizes if the following three conditions hold.
\begin{enumerate}[label=(R\arabic*)]
   \item For every block $j$, exactly $s_j$ vertices $v$ satisfy $b(v)=j$.
   \item No edge has both endpoints in the same block.
   \item Every vertex $v$ has exactly $d_{b(v)}$ neighbours $u$ with $b(u)>b(v)$.
\end{enumerate}

\begin{lemma}\label{lem:order-free}
The values $d_1>\cdots>d_q\ge 0$ and sizes $s_1,\ldots,s_q$ can be realized by an ordering whose blocks are independent if and only if there is a map $b:V\to\{1,\ldots,q\}$ satisfying \textnormal{(R1)--(R3)}. Consequently, by formula~\eqref{eq:block-cost}
\[
   \mathrm{msvc}(G)
   =
   \min\bigl\{
      \mathrm{cost}(d,s):
      \text{there is a map satisfying \textnormal{(R1)--(R3)}}
   \bigr\}.
\]
\end{lemma}

\begin{proof}
Suppose first that such a map $b$ is given. We order the vertices block by block, starting with block $1$ and ending with block $q$, and we order the vertices inside each block arbitrarily. Since no block contains an edge, a vertex has no neighbour in its own block. Therefore its later neighbours in the ordering are exactly its neighbours in blocks with larger index. Condition \textnormal{(R3)} gives the required right degree, and condition \textnormal{(R1)} gives the required block sizes.

Conversely, suppose an ordering realizes the chosen values and sizes and its blocks are independent. Assign each vertex to the block in which it appears. The block sizes give \textnormal{(R1)}, independence gives \textnormal{(R2)}, and the definition of the right degree gives \textnormal{(R3)}.
\end{proof}

There are several choices of block partitions to enumerate. Indeed, the expanded non-increasing right-degree sequence is an integer partition of $m$ into at most $n$ parts, so the Hardy and Ramanujan bound gives $2^{O(\sqrt m)}$ possibilities~\cite{HR18}. Also, because the values $d_1,\ldots,d_q$ are distinct nonnegative integers, $m\ge \sum_{j=1}^q d_j\ge 0+1+\cdots+(q-1)=\binom q2$. Hence every choice has $q=O(\sqrt m)$ non-empty blocks.

It remains to test whether a fixed choice of values and sizes is feasible. This is where the separator structure is used. Fix $d_1>\cdots>d_q\ge 0$ and $s_1,\ldots,s_q$ as above. A subproblem consists of an active vertex set $U$, a boundary set $B$ disjoint from $U$, a block assignment $b_B:B\to\{1,\ldots,q\}$ for the boundary vertices, a residual demand function $t:B\to\{0,\ldots,n\}$, and a vector $\widehat s=(\widehat s_1,\ldots,\widehat s_q)$ of remaining block sizes. We denote such a subproblem by $(U,B,b_B,t,\widehat s)$.

The intended meaning is the following. The vertices in $B$ have already been assigned to blocks. The vertices in $U$ are still active and must be assigned to a block now or in recursive calls. For a boundary vertex $v\in B$, the number $t(v)$ is the number of later neighbours that $v$ still needs to find inside the active set $U$. Finally, $\widehat s_j$ is the number of vertices of $U$ that must be assigned to block $j$.

Formally, the subproblem $(U,B,b_B,t,\widehat s)$ is feasible if $b_B$ can be extended to a map $b:U\cup B\to\{1,\ldots,q\}$ satisfying the following conditions.
\begin{enumerate}[label=(S\arabic*)]
   \item For every block $j$, exactly $\widehat s_j$ vertices of $U$ are assigned to block $j$.
   \item No edge of $G[U\cup B]$ has both endpoints in the same block.
   \item Every active vertex $u\in U$ has exactly $d_{b(u)}$ neighbours $x\in U\cup B$ with $b(x)>b(u)$.
   \item Every boundary vertex $v\in B$ has exactly $t(v)$ neighbours $x\in U$ with $b(x)>b_B(v)$.
\end{enumerate}
The chosen values and sizes are feasible exactly when the root subproblem $(V,\emptyset,\emptyset,\emptyset,(s_1,\ldots,s_q))$ is feasible.

We now describe the recursion. If the active set $U$ is small, say $|U|=O(\sqrt m)$, we solve the subproblem directly by trying all $q^{|U|}$ assignments of the active vertices and checking the four feasibility conditions.

Otherwise, we take a balanced separator $S\subseteq U$ in $G[U]$. Let $A$ and $C$ be the two sides of $U\setminus S$, so there is no edge between $A$ and $C$, and $|A|,|C|\le 2|U|/3$. We guess the block assignment of every vertex in $S$ and then move $S$ to the boundary. Write $B'=B\cup S$ for the new boundary. If the guessed assignment creates an edge inside $G[B']$ whose endpoints are in the same block, then this guess is infeasible and we discard it.

After moving $S$ to the boundary, we update the residual demands. For an old boundary vertex $v\in B$, some of its remaining demand may already be satisfied by vertices of $S$. For a new boundary vertex $v\in S$, its demand is its prescribed right degree minus the later neighbours it already has inside the new boundary. Thus we define
\[
   \rho(v)=
   \begin{cases}
      t(v)-|\{x\in N(v)\cap S : b(x)>b_B(v)\}|, & v\in B,\\[2mm]
      d_{b(v)}-|\{x\in N(v)\cap B' : b(x)>b(v)\}|, & v\in S.
   \end{cases}
\]
If $\rho(v)<0$ for some $v\in B'$, the guess is infeasible. Otherwise, for every $v\in B'$, we guess how much of $\rho(v)$ will be satisfied on the two sides; that is, we choose nonnegative integers $t_A(v)$ and $t_C(v)$ with $\rho(v)=t_A(v)+t_C(v)$.

We also split the remaining block sizes. Let $a_j$ be the number of separator vertices assigned to block $j$. If $a_j>\widehat s_j$ for some $j$, the guess is infeasible. Otherwise, for every block $j$, we split the remaining quantity $\widehat s_j-a_j$ into two nonnegative integers $\widehat s^A_j$ and $\widehat s^C_j$, where $\widehat s_j-a_j=\widehat s^A_j+\widehat s^C_j$.

This produces two recursive subproblems,
\[
   (A,B',b_{B'},t_A,\widehat s^A)
   \qquad\text{and}\qquad
   (C,B',b_{B'},t_C,\widehat s^C).
\]
The current subproblem is feasible if and only if at least one set of guesses makes both recursive subproblems feasible. The reason is that there are no edges between $A$ and $C$. Hence the only interaction between the two sides is through the boundary demands and the block-size constraints, and these are exactly the quantities that we split.

\begin{center}
\noindent\fbox{\parbox{0.97\textwidth}{
\textbf{Algorithm \textsc{Separator-MSVC}$(G)$.}
\begin{enumerate}[label=\arabic*.]
   \item Enumerate all choices $(d_1^{s_1},\ldots,d_q^{s_q})$ with $d_1>\cdots>d_q\ge 0$, $\sum_j s_j=n$, and\ \ \ $\sum_j s_jd_j{=}m$.
   \item Test feasibility by the separator recursion described above.
   \item Return a feasible choice of minimum cost using~\eqref{eq:block-cost}, and output an ordering by listing its blocks in order.
\end{enumerate}
}}
\end{center}

We now bound the running time. Suppose every subgraph of $G$ on $N$ vertices has a balanced separator of size at most $\operatorname{sep}(N)$, where $\operatorname{sep}(.)$ is nondecreasing, and suppose that such separators are given or can be found in polynomial-time. Define
\[
   \Lambda(n)
   :=
   \sum_{i\ge 0}
   \operatorname{sep}\!\left(\left\lceil (2/3)^i n\right\rceil\right).
\]
The quantity $\Lambda(n)$ upper-bounds the number of boundary vertices that can accumulate along any root-to-leaf path of the recursion.

\begin{restatable}{theorem}{thmexactgeneral}\label{thm:exact-general}
\textsc{Minimum Sum Vertex Cover} can be solved exactly in time and space
\[
   2^{O((\Lambda(n)+\sqrt m)\log n)}.
\]
\end{restatable}

\begin{proof}
We first enumerate all choices of right-degree values and block sizes. As argued above, there are $2^{O(\sqrt m)}$ choices, and each has $q=O(\sqrt m)$ blocks. By Lemma~\ref{lem:order-free}, it is enough to test each choice for feasibility.

Fix one such choice. The recursion depth is $O(\log n)$, because each separator step reduces the size of the active set by a constant factor. Along a root-to-leaf path, the number of boundary vertices is at most $\Lambda(n)$ by definition. Therefore, at any node of the recursion tree, a dynamic-programming state is determined by three types of information: the block assignment of the boundary vertices, the residual demand of each boundary vertex, and the remaining block-size vector.

The block assignment contributes at most $q^{\Lambda(n)}$ possibilities. The residual demands contribute at most $(n+1)^{\Lambda(n)}$ possibilities. The block-size vector contributes at most $(n+1)^q$ possibilities. Since $q=O(\sqrt m)$, the number of states for a fixed node and a fixed choice is at most
\[
   q^{\Lambda(n)}(n+1)^{\Lambda(n)}(n+1)^q
   =
   2^{O((\Lambda(n)+\sqrt m)\log n)}.
\]

At a leaf, we test feasibility by enumerating all assignments of the active vertices; this is within the same bound because the leaf size is $O(\sqrt m)$. At an internal node, the transition tries assignments of the separator vertices and all compatible splits of the residual demands and block counts. Equivalently, one may fill the table bottom-up by matching pairs of compatible child states. The cost is polynomial in the number of compatible choices times the table size, and is absorbed by the above bound.

Multiplying by the number of choices does not change the asymptotic expression, because enumerating the candidate block structures contributes only the additional factor $2^{O(\sqrt m)}$. Finally, by storing witnesses in the dynamic program, we can reconstruct a realizing map $b$, and then obtain an optimal ordering by listing the blocks in increasing order.
\end{proof}

\begin{restatable}{corollary}{thmplanarmsvcexact}\label{thm:planar-msvc-exact}
For planar graphs, bounded-genus graphs, and $H$-minor-free graphs, \textsc{Minimum Sum Vertex Cover} can be solved exactly in time and space $2^{O(\sqrt n\log n)}$.
\end{restatable}

\begin{proof}
These graph classes have balanced separators of size $O(\sqrt N)$: for planar graphs this is the theorem of Lipton and Tarjan~\cite{LiptonTarjan79}, and for $H$-minor-free graphs this follows from the theorem of Alon, Seymour, and Thomas~\cite{AlonSeymourThomas90}. Therefore $\Lambda(n)=O(\sqrt n)$. These classes are also sparse, so $m=O(n)$. The result follows from Theorem~\ref{thm:exact-general}.
\end{proof}

\begin{restatable}{corollary}{corexacttreewidth}\label{cor:exact-treewidth}
For graphs of treewidth $t$, \textsc{Minimum Sum Vertex Cover} can be solved exactly in time and space $2^{O((t\log n+\sqrt{tn})\log n)}$. In particular, for constant treewidth the running time is $2^{O(\sqrt n\log n)}$.
\end{restatable}

\begin{proof}
Every graph of treewidth $t$ has a balanced separator of size at most $t+1$, and the same is true for all subgraphs. Hence $\Lambda(n)=O(t\log n)$. Moreover, graphs of treewidth $t$ have $m=O(tn)$ edges. The claim follows from Theorem~\ref{thm:exact-general}.
\end{proof}

\section{Lower Bound for Exact Planar MSVC} \label{sec:planar-lb}

In this section we prove that the separator-based algorithm (Corollary~\ref{thm:planar-msvc-exact}) is  best possible, up to logarithmic factors in the exponent, for planar
graphs.  We show that MSVC remains NP-complete on planar graphs, and that, under
ETH, the running time cannot be improved to $2^{o(\sqrt n)}$.

We start with the following orientation formulation of minimum sum vertex cover. We show that the reduction of Cardinal, Fiorini, and Joret~\cite{CFJ08} developed for the NP-hardness of \emph{minimum entropy orientation} can be used to show that the minimum sum vertex cover of an $n$-vertex planar graph does not admit an exact algorithm with running time~$2^{o(\sqrt{n})}$, unless ETH fails. 

\begin{lemma}
\label{lem:orientation-equivalence}
For an orientation $\vec G$ of $G$, let $d_1(\vec G)\ge d_2(\vec G)\ge\cdots\ge d_n(\vec G)$ be the sorted indegree sequence of $\vec G$. Let $F(\vec G)=\sum_{i=1}^n i\,d_i(\vec G)$. Then $\mathrm{msvc}(G)=\min_{\vec G}F(\vec G)$.
\end{lemma}
\begin{proof}
Let $\sigma^\star$ be an optimal ordering whose right-degree sequence is nonincreasing; such an ordering exists by Observation~\ref{obs:nonincreasing}. Orient every edge toward its earlier endpoint in $\sigma^\star$. Then the indegree of the vertex in position $i$ is exactly its right degree in $\sigma^\star$. Since these right degrees are nonincreasing, the resulting sorted indegree sequence is the same sequence. Hence the orientation has value equal to the MSVC cost of $\sigma^\star$, and therefore
\[
   \min_{\vec G}\sum_{i=1}^n i\,d_i(\vec G)\le \mathrm{msvc}(G).
\]

Conversely, fix an arbitrary orientation $\vec G$. Order the vertices by nonincreasing indegree, breaking ties arbitrarily, and let $\psi(v)$ be the position of vertex $v$ in this ordering. Consider an edge $uv$ oriented into $v$. In the ordering $\psi$, this edge is covered at time $\min\{\psi(u),\psi(v)\}$, which is at most $\psi(v)$. Charging each edge to its head and summing over all edges gives
\[
  \sum_{uv\in E(G)}\min\{\psi(u),\psi(v)\}
  \le
  \sum_{v\in V(G)}\psi(v)d_v
  = F(\vec G).
\]
Minimizing over all orientations gives the reverse inequality, and the lemma follows.
\end{proof}

Let $a=(a_1,\ldots,a_n)$ and $b=(b_1,\ldots,b_n)$ be nonincreasing integer
sequences with the same sum. We say that $a$ \emph{dominates} $b$ if $\sum_{i=1}^t a_i\ge \sum_{i=1}^t b_i$ for every $t$,
and at least one inequality is strict. The MSVC orientation objective has the strict dominance property: if $a$ dominates $b$, then $\sum_{i=1}^n i a_i < \sum_{i=1}^n i b_i$. Indeed, with $A_t=\sum_{i=1}^t(a_i-b_i)$ and using $A_n=0$,
\[
  \sum_{i=1}^n i(a_i-b_i)=-\sum_{t=1}^{n-1} A_t<0.
\]
Thus, as in the minimum entropy orientation reduction of Cardinal, Fiorini,
and Joret \cite{CFJ08}, an orientation whose indegree sequence is more concentrated is strictly better for MSVC. We use the Planar Exact Cover as the reduction source.

\begin{theorem}[Planar Exact Cover]
\label{thm:exact-cover-source}
Let $I=(X,\mathcal S)$ be an instance with $|X|=|\mathcal S|=q$, every set $S\in\mathcal S$ of size $3$, every element $x\in X$ contained in exactly three sets, and planar incidence graph. Deciding whether there is a subfamily $\mathcal S'\subseteq \mathcal S$ that covers every element of $X$ exactly once is NP-complete. Moreover, unless ETH fails, this problem has no $2^{o(\sqrt q)}$ time algorithm.
\end{theorem}

The theorem is exactly Cubic Planar Positive 1-in-3-SAT written in exact-cover language. Moore and Robson prove NP-completeness of this problem~\cite{moore2001hard}. The ETH lower bound together with the Sparsification Lemma rules out $2^{o(N+M)}$ time algorithms for 3-SAT \cite{impagliazzo2001complexity}; Lichtenstein's planarization reduces 3-SAT to Planar 3-SAT with size $O((N+M)^2)$ \cite{lichtenstein1982planar}, and the planar reduction of Mulzer and Rote reduces Planar 3-SAT to Planar Positive 1-in-3-SAT with only linear blow~up~\cite{mulzer2008minimum}. Applying the Moore and Robson cubic planar positive construction yields an equivalent instance of the stated exact-cover problem with $q=O((N+M)^2)$; hence a $2^{o(\sqrt q)}$ time algorithm for the theorem's problem would give a $2^{o(N+M)}$ time algorithm for 3-SAT, contradicting ETH. 

\begin{lemma}[\cite{CFJ08}]\label{lem:cfj-gadget}
Let $I=(X,\mathcal S)$ be an instance of Planar Exact Cover with $|X|=|\mathcal S|=q$ and $q$ divisible by $3$. In polynomial-time one can construct a planar graph $G_I$ with $|V(G_I)|=7q$ and $|E(G_I)|=12q$, such that the following holds. Define the sorted sequence
\[
  T_q:= (\underbrace{4,\ldots,4}_{q}, \underbrace{3,\ldots,3}_{7q/3}, \underbrace{1,\ldots,1}_{q}, \underbrace{0,\ldots,0}_{8q/3}).
\]
Every sorted in-degree sequence of an orientation of $G_I$ is either dominated by $T_q$ or equal to $T_q$, and $T_q$ is realized by an orientation of $G_I$ if and only if $I$ has an exact cover.
\end{lemma}

As an intermediate step, we first show that the decision version of MSVC is NP-complete even on planar graphs. The decision version of MSVC is already known to be NP-hard on general graphs~\cite[Theorem~2.4]{AP25} and~\cite[Theorem~14]{FLT04}.

\begin{restatable}{theorem}{thmnpcomplete}
\label{thm:np-complete}
The decision version of \textsc{Minimum Sum Vertex Cover} is NP-complete on planar graphs.
\end{restatable}

\begin{proof}
Membership in NP is immediate: given a bijection
$\sigma:\{1,\ldots,|V|\}\to V$, its cost can be computed in polynomial-time.

For NP-hardness, reduce from Planar Exact Cover (Theorem~\ref{thm:exact-cover-source}). Let $I=(X,\mathcal S)$ be an instance with $|X|=|\mathcal S|=q$. If $q$ is not divisible by $3$, then $I$ is a trivial instance, since an exact cover would have to partition $q$ elements into sets of size $3$. Assume now that $3$ divides~$q$. Construct the planar graph $G_I$ from Lemma~\ref{lem:cfj-gadget}. Let $\kappa := \sum_{i=1}^{7q} i\,T_{q,i} =  21q^2+6q$. Then planar MSVC instance is $(G_I,\kappa)$.

By Lemma~\ref{lem:orientation-equivalence}, $\OPT(G_I)=\min_{\vec G}F(\vec G)$. By Lemma~\ref{lem:cfj-gadget}, every sorted in-degree sequence of an orientation of $G_I$ is dominated by $T_q$. This implies $F(\vec G) \ge \kappa$ for every orientation $\vec G$, with equality only when the sorted in-degree sequence is exactly $T_q$. Again by Lemma~\ref{lem:cfj-gadget}, $T_q$ is achievable if and only if $I$ has an exact cover. Therefore $I$ has an exact cover if and only if $\OPT_{MSVC}(G_I)\le \kappa$. The reduction is polynomial and the graph $G_I$ is planar. Hence MSVC is NP-hard on planar graphs.
\end{proof}

\begin{restatable}{theorem}{thmeth}
\label{thm:eth}
Unless ETH fails, \textsc{Minimum Sum Vertex Cover} has no $2^{o(\sqrt n)}$ time algorithm on $n$-vertex planar graphs.
\end{restatable}

\begin{proof}
Suppose, for contradiction, that there is an algorithm deciding planar MSVC on $n$-vertex graphs in time $2^{o(\sqrt n)}$. We use it to solve Planar Exact Cover faster than allowed by Theorem~\ref{thm:exact-cover-source}.

Let $I=(X,\mathcal S)$ be a Planar Exact Cover instance with $|X|=|\mathcal S|=q$. If $q$ is not divisible by $3$, answer no. Otherwise construct the graph $G_I$ from Lemma~\ref{lem:cfj-gadget} and the budget $\kappa=21q^2+6q$ as in the proof of Theorem~\ref{thm:np-complete}. The constructed graph has $n=|V(G_I)|=7q$ vertices. Running the assumed planar MSVC algorithm on $(G_I,\kappa)$ decides whether $I$ has an exact cover. The running time is
\[
  2^{o(\sqrt n)}\,q^{O(1)} = 2^{o(\sqrt{7q})}\,q^{O(1)} = 2^{o(\sqrt q)}.
\]
This contradicts the lower bound in Theorem~\ref{thm:exact-cover-source}, unless ETH fails. Therefore planar MSVC has no $2^{o(\sqrt n)}$ time algorithm under ETH.
\end{proof}

\bibliography{main}

\newpage
\appendix

\section{An Exact Algorithm and Lower Bound for General Graphs}
\label{sec:appendix-general}
 
For completeness we record the exact complexity of MSVC on general graphs. We give a Held and Karp~\cite{HeldK61} style subset dynamic program running in $O^\ast(2^n)$ time (Theorem~\ref{thm:dp-upper}) and show that it is qualitatively optimal: under Gap-ETH, MSVC admits no $2^{o(n)}$ time algorithm (Theorem~\ref{thm:msvc-subexp-lb}), through a linear-size reduction from bounded-degree gap Vertex Cover. These two theorems together imply \Cref{thm:exact-general-graphs}.

\subsection{A subset dynamic program}\label{subsec:dp}
 
For $S\subseteq V$, let $u(S)=|E[V\setminus S]|$ be the number of edges with both endpoints outside $S$; if the vertices of $S$ have already appeared as the prefix of an ordering, then $u(S)$ is precisely the number of edges still uncovered. Maintain a table $D(S)$, $S\subseteq V$, equal to the least cost accumulated by an ordering whose first $|S|$ vertices are exactly the vertices of $S$. Initialize $D(\varnothing)=0$ and, for every $S\subseteq V$ and every $v\in V\setminus S$, relax
\[
   D(S\cup\{v\})\ \gets\ \min\bigl\{\,D(S\cup\{v\}),\ D(S)+u(S)\,\bigr\};
\]
the answer is $D(V)$.
 
\begin{theorem}\label{thm:dp-upper}
Let $G$ be an $n$-vertex graph. \textsc{Minimum Sum Vertex Cover} on $G$ can be solved in $2^n n^2$ time and $2^n n$ space.
\end{theorem}
 
\begin{proof}
We claim that the dynamic program computes $D(V)=\mathrm{msvc}(G)$. Consider an ordering $v_1,\ldots,v_n$ and set $S_i=\{v_1,\ldots,v_i\}$ with $S_0=\varnothing$. Immediately before $v_{i+1}$ is chosen, the uncovered edges are exactly those induced by $V\setminus S_i$, so their number is $u(S_i)$. As in Equation~\eqref{eq:svc-uncovered}, the cost of the ordering equals $\sum_{i=0}^{n-1}u(S_i)$. The dynamic program minimizes precisely this sum over all prefixes and all choices of the next vertex, so $D(V)=\mathrm{msvc}(G)$.
 
There are $2^n$ states, each with at most $n$ outgoing transitions. The counts satisfy $u(S\cup\{v\})=u(S)-|N(v)\setminus(S\cup\{v\})|$, so they can computed by a precomputed for all subsets, within the same $O^\ast(2^n)$ bound; evaluating each transition in $O(n)$ time gives $2^n n^2$ time. Storing $D$ together with one argmin vertex per state uses $2^n n$ space, and tracing the argmins from $V$ back to $\varnothing$ recovers an optimal ordering.
\end{proof}
 
\subsection{A subexponential-time lower bound}\label{subsec:genlb}
 
We now show that the algorithm above is qualitatively optimal: under Gap-ETH, MSVC cannot be solved in subexponential time. We use the following bounded-degree gap hardness of Vertex Cover, where $\tau(G)$ denotes the size of a minimum vertex cover of $G$.
 
\begin{theorem}[Bounded-degree gap hardness of Vertex Cover]
\label{thm:bd-vc-gap}
Assuming Gap-ETH, there exist constants $\Delta\ge 3$ and $\gamma>0$ such that no $2^{o(n)}$ time algorithm distinguishes, given an $n$-vertex graph $G$ of maximum degree at most $\Delta$ and an integer $t$, between $\tau(G)\le t$ and $\tau(G)\ge t+\gamma n$.
\end{theorem}
 
This follows from the linear-size, gap-preserving APX-hardness reductions for Minimum Vertex Cover on bounded-degree (cubic) graphs~\cite{AlimontiKann2000,ChlebikChlebikova2006}, composed with Gap-ETH~\cite{Dinur16,ManurangsiR17}; since the reductions are linear and gap-preserving, the additive gap $\gamma n$ and the $2^{o(n)}$ bound follow.
 
\begin{theorem}
\label{thm:msvc-subexp-lb}
Assuming Gap-ETH, \textsc{Minimum Sum Vertex Cover} has no exact algorithm running in $2^{o(N)}$ time on $N$-vertex graphs.
\end{theorem}
 
\begin{proof}
To prove this we use a similar reduction to the one in~\cite[Theorem 14]{FLT04}. Let $(G,t)$ be an instance of bounded-degree gap Vertex Cover from Theorem~\ref{thm:bd-vc-gap}, with $G=(V,E)$, $|V|=n$, $|E|=m$, and $\Delta(G)\le\Delta$, so that $m\le\Delta n/2$. Fix an integer constant $k\ge\Delta/\gamma$ and let $G'$ be the disjoint union of $G$ with a matching of $M:=kn$ edges (that is, we add $M$ isolated edges). The resulting instance has $N=n+2M=(1+2k)n=O(n)$ vertices, so the reduction is linear.
 
Let $\tau(G) \le t$ and $C$ be a vertex cover of $G$ with $|C|\le t$. Place $C$, padded by arbitrary vertices of $G$ to length exactly $t$, in the first $t$ positions; then one endpoint of each matching edge in positions $t+1,\ldots,t+M$; then the remaining vertices. Ordering $C$ uniformly at random covers each edge of $G$ within expected time $(t+1)/2$, so some ordering of the first block makes the edges of $G$ contribute at most $m(t+1)/2$. The matching edges contribute $\sum_{i=1}^M (t+i)=Mt+\tfrac{M(M+1)}2$. Hence
\[
   \mathrm{msvc}(G')\le\frac{m(t+1)}2+Mt+\frac{M(M+1)}2 =: B.
\]
 
Let $\tau(G) \ge t + \gamma n$. Fix an ordering of $G'$. Let the vertices of $G$ that cover at least one new edge of $G$ when placed occupy positions $g_1<\cdots<g_\theta$; since they cover all the edges of $G$, they contain a vertex cover of $G$, so $\theta\ge\tau(G)$. The vertex at $g_r$ covers $c_r\ge 1$ new edges, all at time $g_r$, so the edges of $G$ contribute $\sum_r c_r g_r\ge\sum_r g_r$. Let the matching edges be first covered at the distinct positions $q_1<\cdots<q_M$; they contribute $\sum_s q_s$. In the cost analysis of the $\mathrm{msvc}(G')$, we can assume that $M$ disjoint edges are covered after edges of $G$ in along the ordering. Therefore
\begin{align*}
\mathrm{msvc}(G')
&\ge \sum_r g_r+\sum_s q_s \ge \binom{\theta+1}{2}+\binom{M+1}{2}+\theta M \ge \binom{M+1}{2}+\theta M \\
&\ge \frac{M(M+1)}{2}+\tau(G)M \ge \frac{M(M+1)}{2}+(t+\gamma n)M
=: L.
\end{align*}
 
Subtracting the two bounds and using $M=kn$, $k\ge\Delta/\gamma$, $m\le\Delta n/2$, and $t\le n$,
\[
   L-B=\gamma nM-\frac{m(t+1)}2=\gamma k\,n^2-\frac{m(t+1)}2
   \ge\Delta n^2-\frac{\Delta n(n+1)}4=\frac{3\Delta}4 n^2-O(n).
\]
Thus $L>B$ for all sufficiently large $n$, and an exact MSVC algorithm distinguishes the two cases by comparing $\mathrm{msvc}(G')$ to any threshold strictly between $B$ and $L$. Since $N=O(n)$, a $2^{o(N)}$ time MSVC algorithm would yield a $2^{o(n)}$ time algorithm for bounded-degree gap Vertex Cover, contradicting Theorem~\ref{thm:bd-vc-gap}.
\end{proof}

\end{document}